\documentclass{ifacconf}

\usepackage{graphicx}      
\usepackage{natbib}        

\usepackage{epstopdf}
\usepackage[utf8]{inputenc}	
\usepackage{amssymb}
\usepackage{amsmath}
\usepackage{mathtools}
\usepackage[mathscr]{eucal}
\usepackage{lipsum}
\usepackage{cuted} 

\usepackage{blindtext}

\usepackage{graphicx,keyval,trig,url} 

\usepackage{dsfont}

\usepackage{algorithm}
\usepackage{algpseudocode}      

\usepackage{epstopdf}
\usepackage{bm}   
\usepackage{bbm} 

\usepackage{graphicx}

\usepackage{subfigure}
\usepackage{hhline}
\usepackage{multirow}
\usepackage{array}
\usepackage{savesym}
\savesymbol{AND}

\usepackage{color}
\usepackage{xcolor}

\usepackage{nicefrac}

\usepackage[colorinlistoftodos, textwidth=3cm, shadow]{todonotes}

\newtheorem{defin}{Definition}
\newtheorem{rem}{{Remark}}
\newtheorem{assm}{Assumption}

\newtheorem{proof}{Proof}

\begin{document}

\begin{frontmatter}

	\title{A Topology-Switching Coalitional Control and Observation Scheme with Stability Guarantees \thanksref{footnoteinfo}}

	\thanks[footnoteinfo]{This research was supported by the Spanish Training programme for Academic Staff (FPU17/02653), the MINECO-Spain project DPI2017-86918-R, and the European Research Council (Advanced Research Grant $769051$-OCONTSOLAR).}
	
	\author[First]{Paula Chanfreut} 
	\author[Second]{Twan Keijzer}
    \author[Second]{Riccardo M.G. Ferrari}
    \author[First]{Jose Maria Maestre}
	
	\address[Second]{Delft Center for Systems and Control, Delft University of Technology, Mekelweg 2, 2628 CD, Delft, The Netherlands. \\
		(e-mail: \{t.keijzer, r.ferrari\}@tudelft.nl)}
		
	\address[First]{Department of Systems and Automation Engineering, University of Seville, Camino de los Descubrimientos, Seville, Spain. \\
		(e-mail: \{pchanfreut, pepemaestre\}@us.es)}
	
\begin{abstract}                
	In this paper a coalitional control and observation scheme is presented in which the coalitions are changed online by enabling and disabling communication links. Transitions between coalitions are made to best balance overall system performance and communication costs. Linear Matrix Inequalities are used to design the controller and observer, guaranteeing stability of the switching system. Simulation results for vehicle platoon control are presented to illustrate the proposed method. 
\end{abstract}


\end{frontmatter}

\section{Introduction}
\label{sec:introduction}
In the last years, growing size and complexity of systems has fostered the development of non-centralized strategies \citep{scattolini2009architectures,negenborn2014distributed}. Within this framework, the global system is partitioned into a set of subsystems of manageable size, which are assigned to a set of local controllers. The individual decisions are taken using local information and possibly other subsystem's data when it is shared by their corresponding control entities. Recent works have proposed control strategies where the communication burden is dynamically adjusted, e.g.,  \textit{coalitional} control. In this case, the set of agents is dynamically divided into disjoint communication components, or \textit{coalitions},  i.e.,  clusters of cooperating controllers which operate in a decentralized manner with respect to the rest of the system \citep{fele2014coalitional,fele2017coalitional}. Generally, the benefits of denser communication depend on the level of coupling among subsystems (see \citealt{rawlings2008coordinating}). That is, when the subsystems are weakly coupled, then the controllers' interaction will not significantly alter the solution obtained in a decentralized manner. Within this framework, it is interesting to look at controllers that maintain performance whilst reducing communication and computational demands. This idea is of interest for controlling large-scale systems where the coupling conditions can vary significantly in time. For example, variable controller structures are proposed  in \cite{dorfler2014sparsity} for improving oscillation damping in power systems, and  in \cite{marzband2017distributed} for coordinating distributed energy resources.
\subsection{Main Contributions}
In this paper the coalitional control scheme, with switching topologies, presented in \cite{maestre2014coalitional,Chanfreut2019} is extended to a more general scenario, while keeping all the stability guarantees. Under this framework, the topology of the communication network is dynamically selected from a set of predefined alternatives. A performance index evaluates the suitability of each topology, considering both the expected improvement for the global behaviour and the communication costs involved. In this paper, the states are no longer assumed to be exactly known by the agents, but are estimated from local measurements. This reduces the amount of information that needs to be communicated between agents. Furthermore, a dynamic control law is introduced, allowing to write extended system dynamics including the input. Using properties of combined control and observation scheme, it is proven that the controlled system with switching topologies is stable. An LMI formulation \citep{alamo2006introducing} is used to obtain
both the controller and observer gains for the set of possible communication topologies.

In Section \ref{sec:problem_statement}, the system considered in this paper is introduced, as well as the control objective. The controller and observer structures are introduced in Section \ref{sec:coal_control}, and the condition for switching between topologies is outlined in Section \ref{sec:coal_perf}. The controller and observer design procedure are described in Section \ref{sec:stable_design}. In Section \ref{sec:sim}, numerical results based on the example of a string of platooning cars are provided. Finally, the conclusion is presented in Section \ref{sec:conclusion}.
\section{Problem Formulation}
\label{sec:problem_statement}
In this paper we consider a system of $N$ interconnected agents, which can all have different internal dynamics. The full system dynamics can be represented as
\begin{equation}
\left\{\begin{aligned}
x_\mathcal{N}(k+1)&= A_\mathcal{N}x_\mathcal{N}(k)+ B_\mathcal{N} u_\mathcal{N}(k),\\
y_\mathcal{N}(k) &= C x_\mathcal{N}(k),
\end{aligned}\right.
\label{eq:fullsys}
\end{equation} 
where $x_\mathcal{N}\in\mathbb{R}^{n_x}$, $y_\mathcal{N}\in\mathbb{R}^{n_y}$, and $u_\mathcal{N}\in\mathbb{R}^{n_u}$ contain the stacked states, outputs, and inputs of all agents, respectively. Furthermore, output matrix $C$ is a block-diagonal matrix such that each agent can only measure local states. Matrices $A_\mathcal{N}\in\mathbb{R}^{n_x \times n_x}$ and $B_\mathcal{N}\in\mathbb{R}^{n_x \times n_u}$ can be full.

The $N$ agents are interconnected through a configurable communication network where each communication link can be dynamically activated or deactivated. Then, it possible to restrict or expand the possibilities for communication among the system. Hereafter, we will use $\mathcal{L}$ to denote the set of communication links, i.e., the connections between agents: 
\begin{equation}
\mathcal{L}\subseteq \mathcal{L}^\mathcal{N}=\{\{i,j\}|\{i,j\}\subseteq \mathcal{N}; \ i,j \in \mathcal{N}; \ i \neq j \}.
\end{equation}
\noindent Likewise, $\Lambda$ will denote a set of enabled links describing the network topology, and $\mathcal{T}$ will represent the set of possible topologies, that is,
\begin{equation}
\mathcal{T}= \lbrace \Lambda  \ | \ \Lambda \subseteq \mathcal{L} \rbrace.
\label{eq:topology}
\end{equation}
Note that each topology $\Lambda \in \mathcal{T}$ arranges the set of agents into clusters that can share data, i.e., the coalitions. Also, notice that the number of coalitions will range from one, when all agents are connected (i.e., $\Lambda= \mathcal{L}$), to $N$, corresponding to a fully decentralized system (i.e., $\Lambda= \lbrace \emptyset \rbrace$).

\begin{assm}\label{ass:stability}
In the topology where all agents are connected, $\Lambda= \mathcal{L}$, system \eqref{eq:fullsys} should be detectable and stabilizable.
\end{assm}

By imposing a certain topology of the network connecting the agents, the global system can be split up into disjoint coalitions. In particular, for a given topology,~$\Lambda$, system~\eqref{eq:fullsys} can be rewritten as
\begin{equation}
\left\{\begin{aligned}
x_\mathcal{N}(k+1)&= (A_\Lambda+A_\mathcal{R}) x_\mathcal{N}(k)+(B_\Lambda+B_\mathcal{R}) u_\mathcal{N}(k),\\
y_\mathcal{N}(k) &= C x_\mathcal{N}(k),
\end{aligned}\right.
\label{eq:coalitionsys}
\end{equation}
where $A_\Lambda$ and $B_\Lambda$ contain only the block-diagonal elements that represent the dynamics within the coalitions, and $A_\mathcal{R}=A_\mathcal{N}-A_\Lambda$ and $B_\mathcal{R}=B_\mathcal{N}-B_\Lambda$ represent only the interaction between coalitions. More generally, we will say that matrices $A_\Lambda$ and $B_\Lambda$ belong to set $\mathcal{M}_\Lambda$, which is defined as follows: 

\begin{defin}
Any matrix belong to set  $\mathcal{M}_\Lambda$ if all the entries connecting variables pertaining different coalitions are zero. 
\end{defin}

For a given topology, the following assumptions hold:
\begin{assm}\label{ass:communication}
Within a coalition, the agents communicate their local measurements and input at every time~step.
\end{assm}
\begin{assm}\label{ass:knowdynamics}
Each agent knows the rows pertaining their own coalition, of system dynamics matrices $A_\Lambda$, $B_\Lambda$, and $C$.
\end{assm}

\subsection{Control Objective}
The control objective is to concurrently optimize system performance, and reduce communication costs for all time. Assuming that the control goal is to regulate the state of all agents towards the origin, the objective function at time $k$ can be defined as
\begin{equation}\label{eq:CentralizedProblem}
 J(k)=\sum\limits_{j=0}^{\infty}  \ell(k+j) + c \ |\Lambda(k+j)|,
 \end{equation}
 \noindent where 
 \begin{equation}
 \ell(k) = \left[\begin{matrix}x_\mathcal{N}(k)\\u_\mathcal{N}(k)\\e_{x_\mathcal{N}}(k)\end{matrix}\right]^\top Q_\Lambda \left[\begin{matrix}x_\mathcal{N}(k)\\u_\mathcal{N}(k)\\e_{x_\mathcal{N}}(k)\end{matrix}\right];
\end{equation}
\noindent $e_{x_\mathcal{N}}$ is the state estimation error of the observer, which will be defined in section \ref{sec:coal_control}; $|\Lambda|$ is the cardinality of set $\Lambda$, representing the number of active communication links; $c$ is the cost for enabling a communication link; and $Q_\Lambda\succ0$ is a weighting matrix, which is defined as $Q_\Lambda=\mathrm{diag}(Q_x,R,Q_e)$, where $Q_x$, $R$ and $Q_e$ correspond to the terms of $x_\mathcal{N}$, $u_\mathcal{N}$ and $e_{x_\mathcal{N}}$ respectively.

 Notice that the introduction of variable $\Lambda$ noticeably hinders the control problem. In previous work \citep{maestre2014coalitional,Chanfreut2019} the authors have found an approximate optimal solution to this problem using Linear Matrix Inequalities (LMIs), assuming the agents have perfect knowledge of their state. In this paper this assumption is removed and an alternative control and observation scheme is introduced, while still only requiring LMIs to be solved for a guaranteed stable design.
\section{Coalitional Controller \& Observer}
\label{sec:coal_control} 
In this section the controller and observer forms used in the rest of this paper are defined, thus allowing us to write the combined dynamics of $\left[\begin{matrix}x_\mathcal{N}& u_\mathcal{N}& e_{x_\mathcal{N}}\end{matrix}\right]^\top$. The gain matrices appearing in these dynamics will be defined in Section \ref{ssec:ctrlobsgain}.

A controller will be designed to have the form
\begin{equation}\label{eq:ControlLaw}
    u_\mathcal{N}(k+1)=K_{x\Lambda}\hat{x}_\mathcal{N}(k) + K_{u\Lambda}u_\mathcal{N}(k),
\end{equation}
where $K_{x\Lambda}\in\mathcal{M}_\Lambda$ and $K_{u\Lambda}\in\mathcal{M}_\Lambda$ are gain matrices, and $\hat{x}_\mathcal{N}$ contains the stacked local state estimates, which result from the local observers in each agent.

The observer will be designed to have the form
\begin{equation}\label{eq:Observer}
    \hat{x}_\mathcal{N}(k+1) = A_\Lambda\hat{x}_\mathcal{N}(k) + B_\Lambda u_\mathcal{N}(k) + L_\Lambda(y_\mathcal{N}(k)-C\hat{x}_\mathcal{N}(k)),
\end{equation}
where $L_\Lambda\in\mathcal{M}_\Lambda$ is the observer gain matrix. The observer estimation error is defined as $e_{x_\mathcal{N}}=\hat{x}_\mathcal{N}-x_\mathcal{N}$.

\begin{rem}
  Equations \eqref{eq:ControlLaw} and \eqref{eq:Observer} define controller and observer forms in which control and observation are decoupled per coalition. 
\end{rem}

The combined dynamics of system, controller and observer can now be described by
\begin{equation}\label{eq:ControlledSys}\small{
\left[\begin{matrix}x_\mathcal{N}(k+1)\\u_\mathcal{N}(k+1)\\e_{x_\mathcal{N}}(k+1)\end{matrix}\right] = \left[\begin{matrix}
A_\mathcal{N}&B_\mathcal{N}&0\\
K_{x\Lambda}&K_{u\Lambda}&K_{x\Lambda}\\
A_\Lambda-A_\mathcal{N}&B_\Lambda-B_\mathcal{N}&A_\Lambda+L_\Lambda C
\end{matrix}\right]
\left[\begin{matrix}x_\mathcal{N}(k)\\u_\mathcal{N}(k)\\e_{x_\mathcal{N}}(k)\end{matrix}\right],}
\end{equation}
which can, in simplified notation, be rewritten to
\begin{equation}\label{eq:ControlledSysShort}
    \xi(k+1)=\mathbb{A}\xi(k),
\end{equation}
where $\xi$ and $\mathbb{A}$ are defined from equation \eqref{eq:ControlledSys}.
\section{Topology Switching Conditions}
\label{sec:coal_perf}
 In this section the criterion for evaluating the topologies performance is presented.

Recall the performance function in equation \eqref{eq:CentralizedProblem} and let us consider that there exists a symmetric positive definite matrix  $P_{\Lambda}\in\mathcal{M}_\Lambda$ such that

  \begin{equation}\label{eq:PA}
\xi(k)^\top P_\Lambda \xi(k) \geq \sum_{j=0}^\infty \ell(k+j),
\end{equation}

Note that actual $x_\mathcal{N}$ and $e_{x_\mathcal{N}}$ are not known by the agents, hence this function $\xi(k)^\top P_\Lambda \xi(k)$ cannot be evaluated on-line. However, it is possible to compute its largest possible value, i.e., 
\begin{equation}\label{eq:PA2}
\left[\begin{matrix}\bar{x}_\mathcal{N}(k)\\u_\mathcal{N}(k)\\\bar{e}_{x_\mathcal{N}}(k)\end{matrix}\right]^\top P_\Lambda \left[\begin{matrix}\bar{x}_\mathcal{N}(k)\\u_\mathcal{N}(k)\\\bar{e}_{x_\mathcal{N}}(k)\end{matrix}\right] \geq \sum_{j=0}^\infty \ell(k+j),
\end{equation}
where $\bar{x}_\mathcal{N}$ and $\bar{e}_{x_\mathcal{N}}$  are further defined in Section \ref{ssec:topologycond}. The resulting approximation can be calculated on-line under the following assumption:
\begin{assm}
\label{ass:topologyswitch}
When a topology switch is considered, all communication links are enabled, such that all agents can share their local state estimate and input.
\end{assm}

Regarding the communication costs in \eqref{eq:CentralizedProblem}, the infinite sum is replaced by a parameter $\kappa$, indicating the number of time steps over which the communication costs are counted. Taking both terms together, we consider the following cost function to evaluate topology performance
\begin{equation} \label{eq:r}
r_\Lambda(k)= \left[\begin{matrix}\bar{x}_\mathcal{N}(k)\\u_\mathcal{N}(k)\\\bar{e}_{x_\mathcal{N}}(k)\end{matrix}\right]^\top P_\Lambda \left[\begin{matrix}\bar{x}_\mathcal{N}(k)\\u_\mathcal{N}(k)\\\bar{e}_{x_\mathcal{N}}(k)\end{matrix}\right] + \kappa c |\Lambda|,
\end{equation}
which can, in simplified notation, be rewritten to
\begin{equation}\label{eq:rSimple}
r_\Lambda = \bar{\xi}(k)^\top P_\Lambda \bar{\xi}(k) + \kappa c |\Lambda|.
\end{equation}

In the proposed method, a topology switch is considered at regular intervals of $T$ time-steps. A switch will occur if there exists a topology $\Lambda_j\in\mathcal{T}$, other than the current topology $\Lambda_i$, such that $r_{\Lambda_j}(k)<r_{\Lambda_i}(k)$.
\section{A design method with stability guarantees} 
\label{sec:stable_design}
In this section a procedure is presented for designing the controller and observer gains. Furthermore, the worst case state and state estimation error, required for the topology switching condition, will be presented. The control scheme will then be summarized in a pseudo-code and lastly, we will prove that the proposed control scheme is stable.

\subsection{Design of the Controller and Observer gains}
\label{ssec:ctrlobsgain}
First, the controller gains are designed such that equation \eqref{eq:PA} holds, assuming $e_{x_\mathcal{N}}=0$. For this design objective, the method introduced in \cite{Chanfreut2019,maestre2014coalitional} can be adapted. For this purpose we rewrite equation \eqref{eq:ControlledSys} as
\begin{equation}\label{eq:IdealControlledSys}
\left[\begin{matrix}x_\mathcal{N}(k+1)\\u_\mathcal{N}(k+1)\end{matrix}\right] = \left(\underbrace{\left[\begin{matrix}
A_\mathcal{N}&B_\mathcal{N}\\
0&0
\end{matrix}\right]}_{\mathcal{A}_\mathcal{N}}+\underbrace{\left[\begin{matrix}
0\\
I
\end{matrix}\right]}_{\mathcal{B}_\mathcal{N}}
\underbrace{\left[\begin{matrix}
K_{x\Lambda}&K_{u\Lambda}
\end{matrix}\right]}_{\mathcal{K}_\Lambda}\right)
\left[\begin{matrix}x_\mathcal{N}(k)\\u_\mathcal{N}(k)\end{matrix}\right],
\end{equation}
and equation \eqref{eq:PA} as
\begin{equation}\label{eq:IdealPA}
\left[\begin{matrix}x_\mathcal{N}(k)\\u_\mathcal{N}(k)\end{matrix}\right]^\top \mathcal{P}_\Lambda \left[\begin{matrix}x_\mathcal{N}(k)\\u_\mathcal{N}(k)\end{matrix}\right] \geq \sum_{j=0}^\infty \left[\begin{matrix}x_\mathcal{N}(k+j)\\u_\mathcal{N}(k+j)\end{matrix}\right]^\top \mathcal{Q}_\Lambda \left[\begin{matrix}x_\mathcal{N}(k+j)\\u_\mathcal{N}(k+j)\end{matrix}\right],
\end{equation}
where $\mathcal{P}_\Lambda\in \mathcal{M}_\lambda$, and $\mathcal{Q}_\Lambda=\mathrm{diag}(Q_x,R)$.
\begin{theorem}\label{thm:ControlLMI}
For the controlled system \eqref{eq:IdealControlledSys}, if there exist matrices
\begin{subequations}
\begin{equation}
W_{\Lambda}=\mathcal{P}_{\Lambda}^{-1} \in \mathcal{M}_\Lambda, \quad Y_{\Lambda}=\mathcal{K}_{\Lambda} W_{\Lambda} \in \mathcal{M}_\Lambda,
\end{equation}
\label{eq:change_variable}
\end{subequations}
\noindent obtained from problem
\begin{subequations}
\begin{equation*}
\begin{aligned}
&\max_{W_\Lambda, Y_\Lambda}\text{trace}(W_\Lambda)\\
&\text{subject to}
\phantom{\hspace{5cm}}
\end{aligned}
\end{equation*}
\begin{equation}\label{eq:LMI2a}
\left[ 
\begin{array}{cccc}
W_{\Lambda} & * & *\\
\mathcal{A}_\mathcal{N}W_{\Lambda} + \mathcal{B}_\mathcal{N}Y_{\Lambda} & W_{\Lambda}& * \\
\mathcal{Q}_\mathcal{N}^{1/2}W_{\Lambda} & 0 & I  \\
\end{array} \right] \succ 0.
\end{equation}
\label{eq:LMIS}
\end{subequations}
\noindent Then, the following properties hold:
\begin{itemize}
\item Controller \eqref{eq:ControlLaw} with gain $\mathcal{K}_{\Lambda}$ guarantees closed-loop stability of the ideal system \eqref{eq:IdealControlledSys} and $\mathcal{P}_{\Lambda}$ satisfies \eqref{eq:IdealPA}.
\item The communication constraints imposed by the topology are respected.
\end{itemize}
\end{theorem}
\begin{proof}
The proof follows from Theorem 1 in \cite{maestre2014coalitional}.
\end{proof}

The observer design is based on \cite{BenChabane2014-1}, where the actual state is guaranteed to be within an ellipsoid around the state estimate. To this end, rewrite equation \eqref{eq:coalitionsys} as
\begin{equation}
x_\mathcal{N}(k+1) = A_\Lambda x_\mathcal{N}(k) + B_\Lambda u_\mathcal{N}(k) + \underbrace{\left[\begin{matrix}
A_\mathcal{R}&B_\mathcal{R}
\end{matrix}\right]}_{E}
\underbrace{\left[\begin{matrix}x_\mathcal{N}(k)\\u_\mathcal{N}(k)\end{matrix}\right]}_{\omega(k)},
\end{equation}
and define ellipsoidal set $\mathcal{E}(P,\hat{x},\rho) = \{x:(x-\hat{x})^\top P (x-\hat{x})\leq\rho\}$. 
\begin{theorem}
Assume $x_\mathcal{N}(k)\in\mathcal{E}(\mathbb{P} ,\hat{x}_\mathcal{N}(k),\rho_\Lambda(k))$, where $\mathbb{P} \succ 0$, and assume $\omega(k)\in\mathcal{V}_\mathbf{o}$, where $\mathcal{V}_\mathbf{o}$ is a hyper-cube enclosing $\omega(k)$ for all $k$. If there exist a matrix $\mathbb{Y}_\Lambda = \mathbb{P} L_\Lambda \in\mathbb{R}^{n_x\times n_y}$ satisfying the communication constraints (i.e., $\mathbb{Y}_\Lambda \in \mathcal{M}_\Lambda$); a scalar $\beta_\Lambda \in(0,1)$; and a positive scalar $\sigma_\Lambda$ such that LMI
\begin{equation}
    \left[\begin{matrix}\beta_\Lambda \mathbb{P}&*&*\\
    0&\sigma_{\Lambda}&*\\
    \mathbb{P} A_\Lambda-\mathbb{Y}_\Lambda C &\mathbb{P} E\omega(k)&\mathbb{P}\\
    \end{matrix}\right]\succ 0,
\label{eq:observerlmi1}
\end{equation}
\noindent is satisfied for all $\omega(k)\in\mathcal{V}_\mathbf{o}$, then, the next system state $x_\mathcal{N}(k+1)$ is guaranteed to belong to ellipsoid $\mathcal{E}(\mathbb{P},\hat{x}_\mathcal{N}(k+1),\rho_\Lambda(k+1))$, where
\begin{equation}\label{eq:rho_update}
    \rho_\Lambda(k+1)\leq \beta_\Lambda \rho(k)+\sigma_\Lambda.
\end{equation}
\end{theorem}
\begin{proof}
The proof follows from \cite{BenChabane2014-1}. The following changes have been made here, that do not affect the proof:
\begin{itemize}
    \item A constraint has been imposed on the structure of $\mathbb{Y}_\Lambda$, i.e., $\mathbb{Y}_\Lambda \in \mathcal{M}_\Lambda$.
    \item Vector $\omega(k)$ is of lenght $n_x+n_u$ instead of $n_x+n_y$
\end{itemize}
Both changes influence the feasible set corresponding to the problem, but the proof presented in \cite{BenChabane2014-1} holds without change.
\end{proof}

Variable $\rho_\Lambda(k)$ will converge to $\sigma_\Lambda/(1-\beta_\Lambda)$ for $k\rightarrow\infty$, hence, reducing the size of the ellipsoid can
be done by minimizing a combination of $\sigma_\Lambda$ and $\beta_\Lambda$. (see \cite{BenChabane2014-1}). In particular, to obtain variables  $\mathbb{Y}_\Lambda$, $\beta_\Lambda$ and $\sigma_\Lambda$, we have minimized function $\sigma_\Lambda + \epsilon \beta_\Lambda$, where $\epsilon$ is a weighting parameter, subject to LMI \eqref{eq:observerlmi1}.

\begin{rem}
Hypercube $\mathcal{V}_\mathbf{o}$ is a convex polytope, and $\omega(k)$ appears affine in LMI \eqref{eq:observerlmi1}. Therefore, if the LMI is feasible at the corner points of $\mathcal{V}_\mathbf{o}$, it is feasible for all $\omega(k)$.
\end{rem}
\subsection{Analysis of System Stability}
\label{ssec:stabanalysis}
As the controller and observer are designed independently, stability and performance of the controlled system is not guaranteed. Therefore an analysis LMI is constructed that guarantees the decrease of $\xi(k)^\top P_\Lambda \xi(k)$ between topology switches.

This analysis LMI is equivalent to the one presented in equation \eqref{eq:LMI2a} for design. However, as the control and observer gains are now constants, a simpler form of the LMI can be used, which is shown in equation \eqref{eq:analysisLMI}. (see \cite{maestre2014coalitional})
\begin{equation}
\label{eq:analysisLMI}
\begin{aligned}
\mathbb{A}^\top P_\Lambda \mathbb{A}-P_\Lambda\prec-Q_\Lambda
\end{aligned}
\end{equation}

\subsection{Topology Switch Condition}
\label{ssec:topologycond}
The topology switch criteria presented in equation \eqref{eq:rSimple}, needs to be evaluated by each agent when a topology switch is considered. For this, each agent needs to calculate the state and estimation error, $\bar{x}_\mathcal{N}$ and $\bar{e}_{x_\mathcal{N}}$, that cause the largest cost in equation \eqref{eq:rSimple}. This is equivalent to finding values of $x_\mathcal{N}$ and $e_{x_\mathcal{N}}$ that maximize $\xi(k)^\top P_\Lambda \xi(k)$. 

Due to real-time constraints, it was chosen not to use an optimization problem to obtain the tightest ellipsoid fitting this requirement. Alternatively, a much less computationally expensive, but slightly conservative, method has been chosen. A hyper-cube, $\mathcal{V}_{\mathbf{b},\Lambda}$, is defined, enclosing ellipsoid $\mathcal{E}(\mathbb{P},\hat{x}_\mathcal{N}(k),\rho_\Lambda(k))$. The vertices of this hyper-cube define the $2^{n_x}$ options for $x_\mathcal{N}$ and $e_{x_\mathcal{N}}=\hat{x}_\mathcal{N}-x_\mathcal{N}$ that could be $\bar{x}_\mathcal{N}$ and $\bar{e}_{x_\mathcal{N}}$. So, for all vertices $\xi(k)^\top P_\Lambda \xi(k)$ is evaluated, and $\bar{\rho}(k) = \max_{\mathcal{V}_{b,\Lambda}}(\xi(k)^\top P_\Lambda \xi(k))$ and $\bar{x}_\mathcal{N}$, and $\bar{e}_{x_\mathcal{N}}$ are the arguments that create $\bar{\rho}(k)$.

When a topology switch is considered equation \eqref{eq:rSimple} is considered for all possible topologies, and the topology that minimises the criteria is chosen, together with the corresponding control and observer gains.
\subsection{Control scheme}\label{ssec:controlscheme}
In this section the control scheme will be presented that combines all methods outlined in the previous sections. Note that in step 12: $^T$ indicates to the power $T$, not transpose.
\begin{algorithm}
\textbf{Initialize:} Set the initial values for $\hat{x}_{\mathcal{N}}$, $\Lambda$, fix the constants $Q_\Lambda$, $c$, $T$, and determine ellipsoid $\mathcal{E}(\mathbb{P},\hat{x}_\mathcal{N}(0),\rho_\Lambda)$.
\begin{algorithmic}[1]
\ForAll{$\Lambda\in\mathcal{T}$} \Comment{Offline Design}
\State Design $K_{x\Lambda}$ and $K_{u\Lambda}$ \Comment{eq. \eqref{eq:LMIS}}
\State Design $L_{\Lambda}$, $\beta_\Lambda$ and $\sigma_\Lambda$  \Comment{eq. \eqref{eq:observerlmi1}}
\State Check stability and determine $P_\Lambda$ \Comment{eq. \eqref{eq:analysisLMI}}
\If{Topology $\Lambda$ is infeasible} \State $\mathcal{T}\gets\mathcal{T}\setminus \Lambda$
\EndIf
\EndFor
\While{True} \Comment{Online Execution}
\If{$k/T \in \mathbb{N}$} All agents share local state estimate and input over the whole network
\ForAll{$\Lambda\in\mathcal{T}$}
\State All agents calculate $\rho_\Lambda \gets \beta_\Lambda^{T} \rho_\Lambda + \sum_{i=0}^{T-1} \beta_\Lambda^{i}\sigma_\Lambda$ \Comment{eq. \eqref{eq:rho_update}}
\State Find $\bar{x}_\mathcal{N}(k)$, $\bar{e}_{x_\mathcal{N}}(k)$ and $\bar{\rho}_\Lambda(k)$ \Comment{sec. \ref{ssec:topologycond}}
\EndFor
\State Pick $\Lambda$ for which $r_{\Lambda}(k)$ is minimal. \Comment{eq. \eqref{eq:rSimple}}
\Else $~$All agents share local measurement and input only within their coalition.
\EndIf
\State All agents implement $u_\mathcal{N}(k)$
\State All agents calculate $\hat{x}_\mathcal{N}(k+1)$ \Comment{eq. \eqref{eq:Observer}}
\State All agents calculate $u_\mathcal{N}(k+1)$ \Comment{eq. \eqref{eq:ControlLaw}}
\State $k=k+1$
\EndWhile
\end{algorithmic}
\end{algorithm}
\subsection{Stability Properties}
\label{ssec:stabproof}
\begin{theorem}
The control scheme presented in section \ref{ssec:controlscheme} leads to a controlled system \eqref{eq:ControlledSysShort}, which is guaranteed to be stable.
\end{theorem}
\begin{proof}

Using the analysis LMI from equation \eqref{eq:analysisLMI} in steps 4-7 of the control scheme, it is guaranteed that $r_\Lambda$ is decreasing between topology switches for all $\Lambda$ that are left in $\mathcal{T}$. Therefore, it only needs to be proven that the behaviour at topology switches does not destabilize the system.

In \cite{maestre2014coalitional}, the authors prove that index $r_\Lambda$ is a decreasing function for the case in which the real state is known, i.e., there is no estimation error. The imperfect state estimation can cause  temporary increases of $\xi(k)^T P_\Lambda \xi(k)$ due to the switchings between topologies. However, this potential increase is bounded by $\bar{\rho}(k)$, as defined in section~\ref{ssec:topologycond}. Furthermore, as a direct consequence of the switching condition, $\bar{\rho}(k)$ will always stay equal or decrease over every switch. Therefore, this effect cannot destabilize the system.

\end{proof}
\section{Simulation for Vehicle Platoon Control}
\label{sec:sim}
In this section, we apply the coalitional scheme to a dynamic system composed of four coupled agents, i.e., $\mathcal{N}=\lbrace 1,...,4 \rbrace$, which represent four cars. Dynamically, each of the cars is modelled by:
\begin{equation}\label{car_model}
\left[ \begin{array}{c} \dot{p}_i(t) \\ \dot{v}_i(t) \\ \dot{a}_i(t) \end{array} \right] = \left[ \begin{array}{c} v_i(t) \\ a_i(t) \\ \dfrac{1}{\tau}(u_i(t)-a_i(t)) \end{array} \right], 
\end{equation}
\noindent where $p_i$, $v_i$ and $a_i$ denote respectively position, velocity and acceleration of car $i$, $u_i$ is its input, and $\tau$ is a time constant that represents the engine's dynamics. 
\begin{figure}[htt]
\centering
    {\includegraphics[scale=0.41,trim={0cm 0.3cm 0cm 0cm},clip]{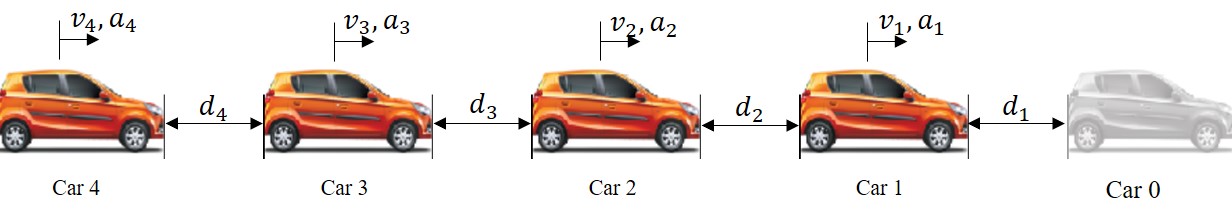}}
	\caption{System of 4 platooning cars following a lead vehicle.}
	\label{sketch_system}
\end{figure}

We consider that the five cars aim at maintaining a desired relative distance 
\begin{equation}
d_{r,i}(t)=r+h v_i(t),
\end{equation}
\noindent where $r$ is the standstill distance and $h$ is the time headway. See Fig. \ref{sketch_system} for an illustration of the system. 

As described in \cite{Zhu2019}, we can derive the following model from \eqref{car_model}:
\begin{equation}\label{eq:model}
\begin{bmatrix} \Delta \dot{d}_i \\ \Delta \dot{v}_i \\ \dot{a}_i \\ \dot{a}_{i-1} \end{bmatrix} =\begin{bmatrix} 0 & 1 & -h & 0 \\ 0 & 0 & -1 & 1 \\ 0 & 0 & -1/\tau  & 0 \\ 0 & 0 & 0 & -1/\tau \end{bmatrix} \begin{bmatrix} \Delta d_i \\ \Delta v_i \\ a_i \\ a_{i-1} \end{bmatrix} + \begin{bmatrix} 0 & 0  \\ 0 & 0  \\ 1/\tau  & 0 \\ 0 & 1/\tau \end{bmatrix}  \begin{bmatrix} u_i \\ u_{i-1} \end{bmatrix},
\end{equation}
where $\Delta d_i (t)= d_i(t)-d_{r,i}(t)$ is the intervehicle distance,  $\Delta v_i(t)=v_{i-1}(t)-v_i(t)$ is the relative velocity between consecutive cars, and $\Delta \dot{v}_i(t) = a_{i-1}(t)-a_i(t)$ is the relative acceleration. Note that the coupling between cars, in \eqref{eq:model}, is through the acceleration and input. The parameters of \eqref{eq:model} used in the simulations are $h=0.7$ and $\tau=0.1s$. Hereon, to better resemble a vehicle platoon scenario, we assume there exists a "car 0", preceding the series of our four cars, for which $a_0=0$ and $u_0=0$. In this case, the aggregation of \eqref{eq:model} for the four cars can be posed in the form of \eqref{eq:fullsys}, where each subsystem state is $x_i=[\Delta d_i, \Delta v_i, a_i]^\mathrm{T}$. The continuous-time dynamics are discretized using zero-order hold and a sampling time of 0.01 s. The control objective here is to regulate all states to zero while minimizing communication costs. In this respect, the stage performance function is defined by weighting matrices $
Q_{x}= I_{n_x}$, $R= 0.1 I_{n_u}$ and $
Q_e= 0.1 I_{n_x}$, where $I_n$ is the identity matrix of dimension $n \times n$. Also, the cost per enabled link $c$ has been set at $2.4$, $\kappa=1 s$, and the value of $\epsilon$ is 0.01.  

The underlying communication network contains three links that connect the cars in series, i.e.,  $\mathcal{L}=\{\text{I,II,III}\}$. In this respect, we assume that an enabled link allows the flow of information in both directions; also, we consider that the agents can communicate if they are either directly or indirectly connected by a path of enabled links. In Table~\ref{topologies}, we specify the 8 network topologies that can be imposed on the network.

\begin{table}[b]
\caption{Network topologies for the four cars system.}
\centering
\begin{tabular}{|c|c||c|c|}
\hline 
$\Lambda_0$ & $\lbrace \emptyset \rbrace$ & $\Lambda_4$ & $\lbrace  \text{I, II} \rbrace$ \\ 
\hline 
$\Lambda_1$ & $\lbrace \text{I} \rbrace$ & $\Lambda_5$ & $\lbrace  \text{I,III} \rbrace$ \\ 
\hline 
$\Lambda_2$ & $\lbrace  \text{II} \rbrace$ & $\Lambda_6$ & $\lbrace  \text{II,III} \rbrace$ \\ 
\hline 
$\Lambda_3$ & $\lbrace \text{III} \rbrace$ & $\Lambda_7$ & $\lbrace  \text{I,II,III} \rbrace$ \\ 
\hline 
\end{tabular} 
\label{topologies}
\end{table}

The control scheme described in Section \ref{ssec:controlscheme} has been implemented from initial state $x_{1,0}=[5,4.5,2.5]^\mathrm{T}$, $x_{2,0}=[-3.5,-4,-4.5]^\mathrm{T}$, $x_{3,0}=[2,2.5,3]^\mathrm{T}$, and $x_{4,0}=[-2.5,-3,-2]^\mathrm{T}$. Furthermore, topology switches are considered each $T=1 s$. For the controller and observer design we used the \texttt{Matlab LMI Toolbox} with the \texttt{mincx} solver. Regarding the observer, each car is assumed to measure the distance and relative velocity with respect to the preceding car, as well as their own velocity. From the distance and local velocity each vehicle can calculate the distance error. The acceleration $a_i$ should be estimated. We have set an initial estimation error of $a_i$ to $10 [m/s^2]$. Note that this initial error determines the initial value of $\rho_\Lambda$ introduced in the pseudo-code of Section \ref{ssec:controlscheme}. 

In Fig.~\ref{fig:d_i}, we show the evolution of the distance error for each car, where it can be seen that after five seconds, all cars approximately reach their desired positions. Also, in Fig. \ref{fig:v_i} and Fig. \ref{a_i}, the evolution of the relative speed, and accelerations are illustrated. All figures compare the state behavior of the coalitional controller with the respective centralized and decentralized cases, illustrating how the coalitional approach closely matches the centralized results. Fig. \ref{topologies} shows the sequence of selected topologies over time.
\begin{figure}[ht]
\centering
    {\includegraphics[width=9.5cm,height=7.3cm,trim={0.5cm 0cm 0cm 0cm},clip]{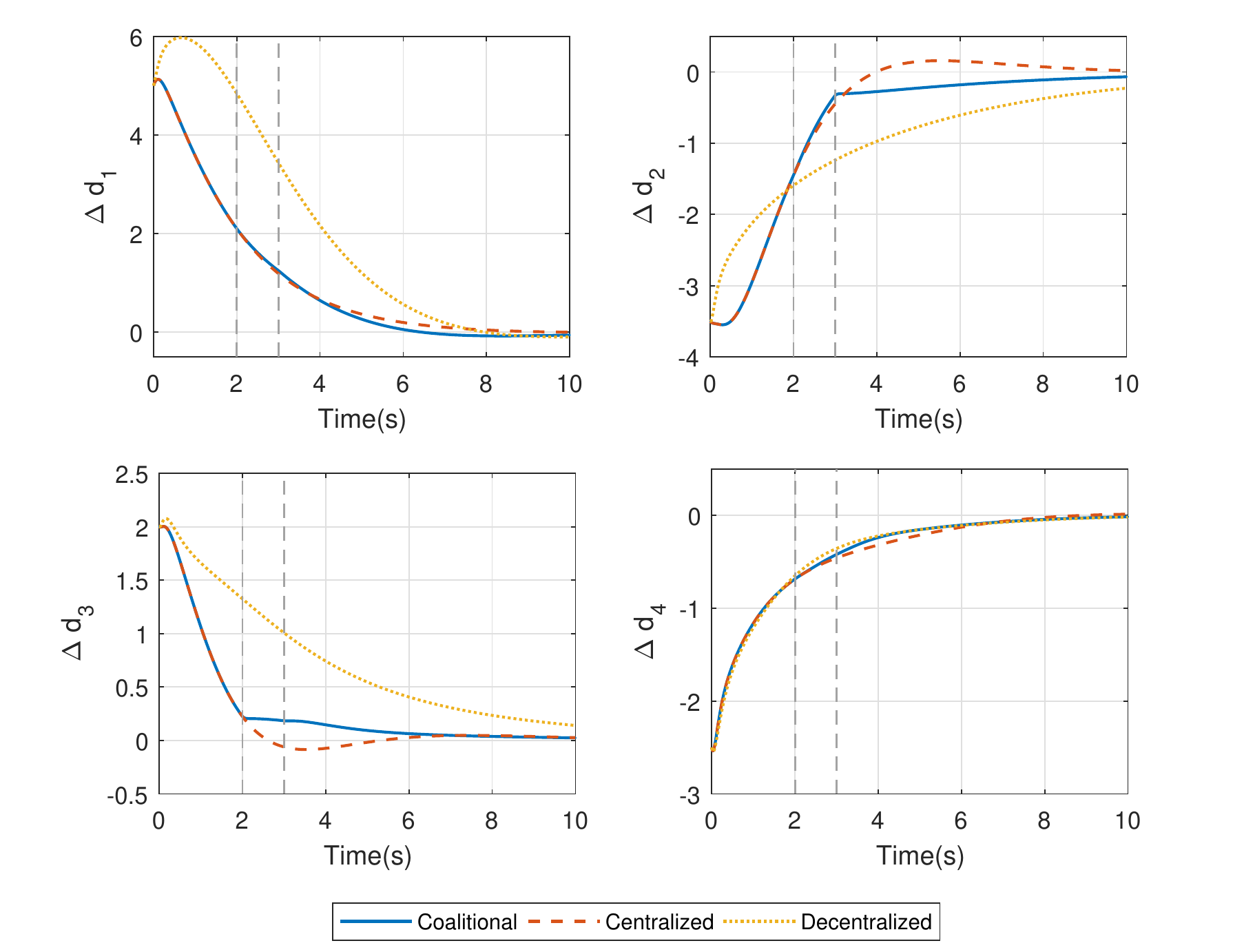}}
    \caption{Evolution of the distance error for each car.}
    \label{fig:d_i}
\end{figure}

\begin{figure}[ht]
\centering
    \includegraphics[width=9.5cm,height=7.3cm,trim={0.5cm 0cm 0cm 0cm},clip]{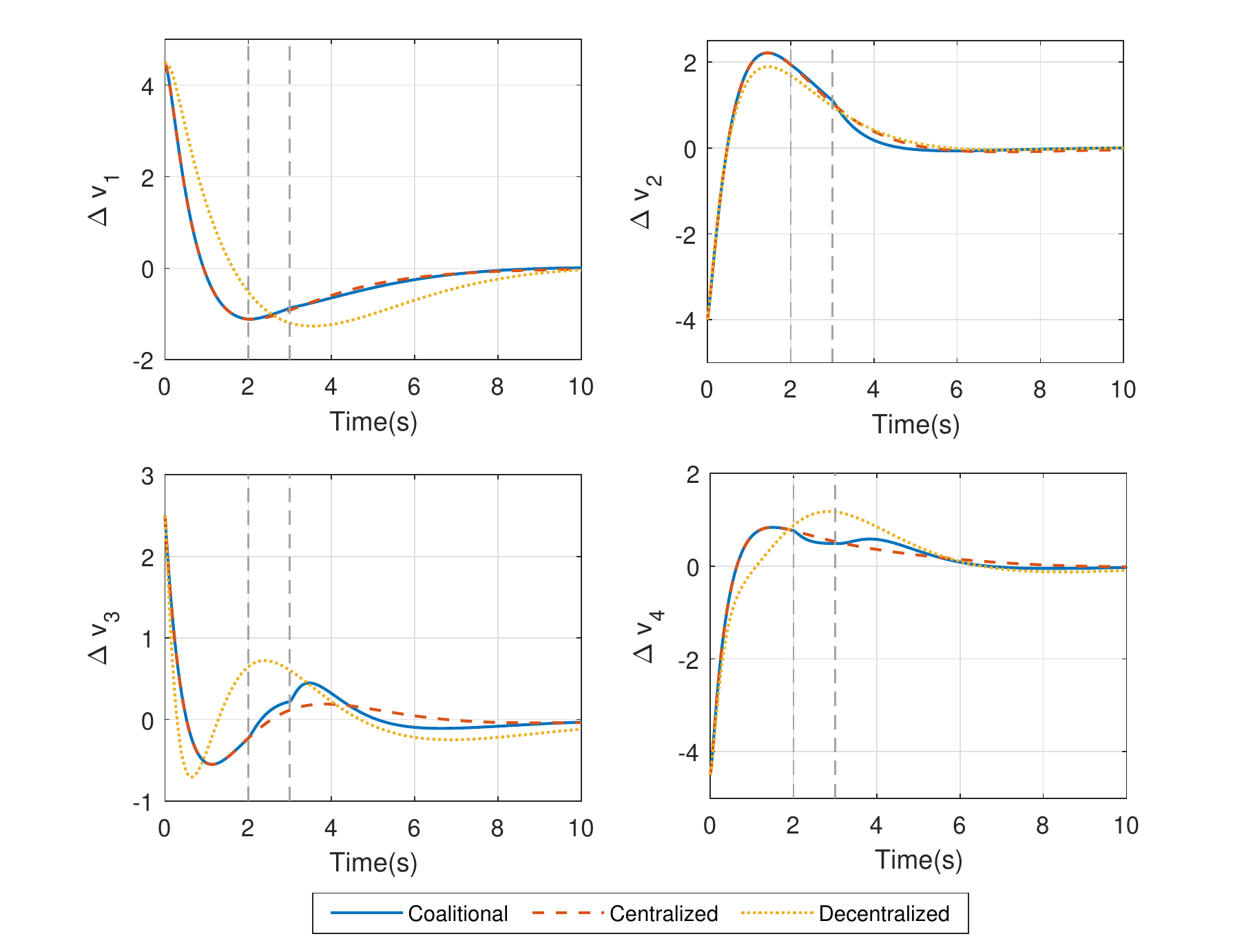}
	\caption{Evolution of the relative speed for each car.}
	\label{fig:v_i}
\end{figure}
\begin{figure}[ht]
\centering
    \includegraphics[width=9.5cm,height=7.3cm,trim={0.5cm 0cm 0cm 0cm},clip]{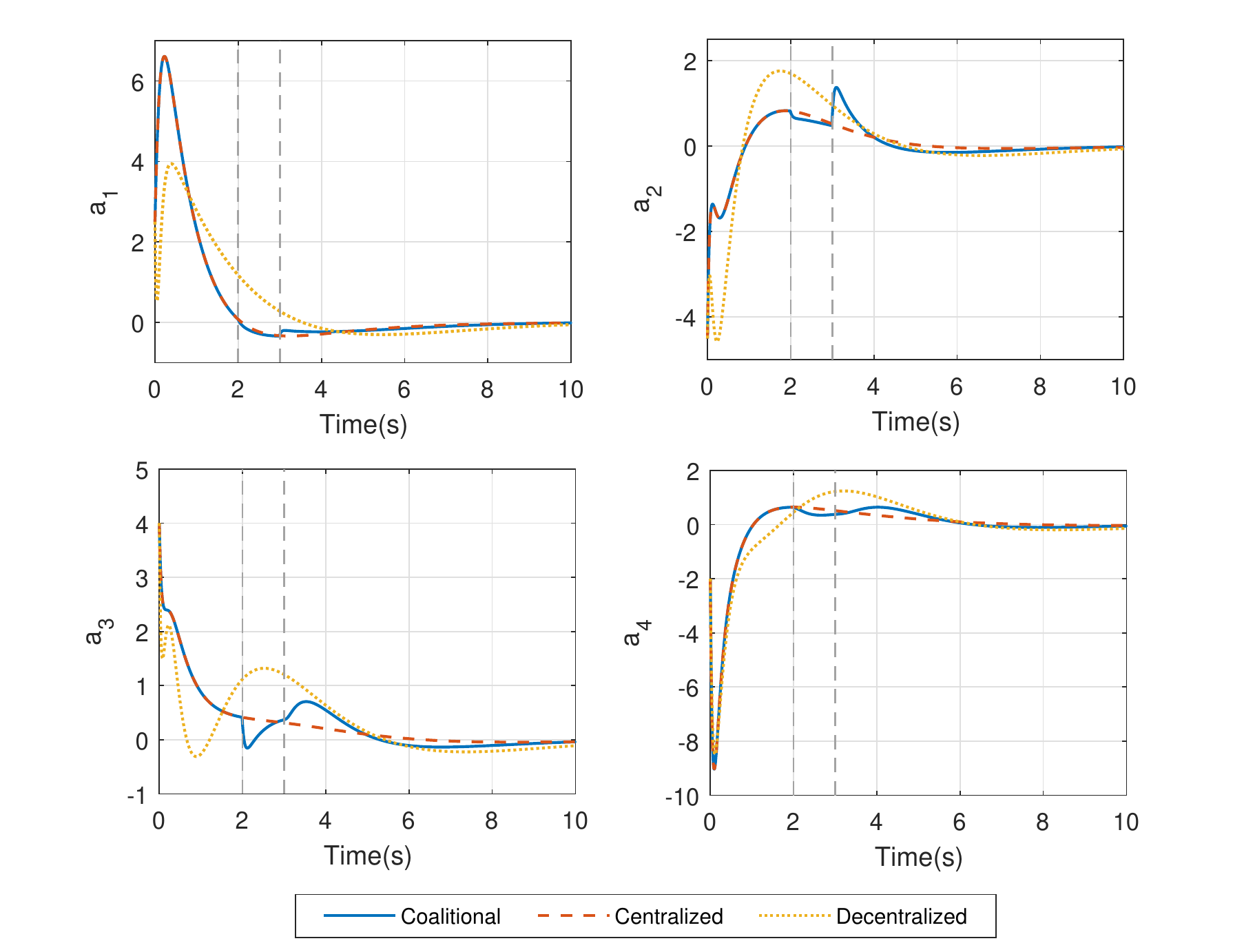}
	\caption{Evolution of the acceleration for each car.}
	\label{a_i}
\end{figure}
\begin{figure}[ht]
\centering
    {\includegraphics[width=7.8cm,height=2.2cm,trim={0.cm 0cm 0cm 0cm},clip]{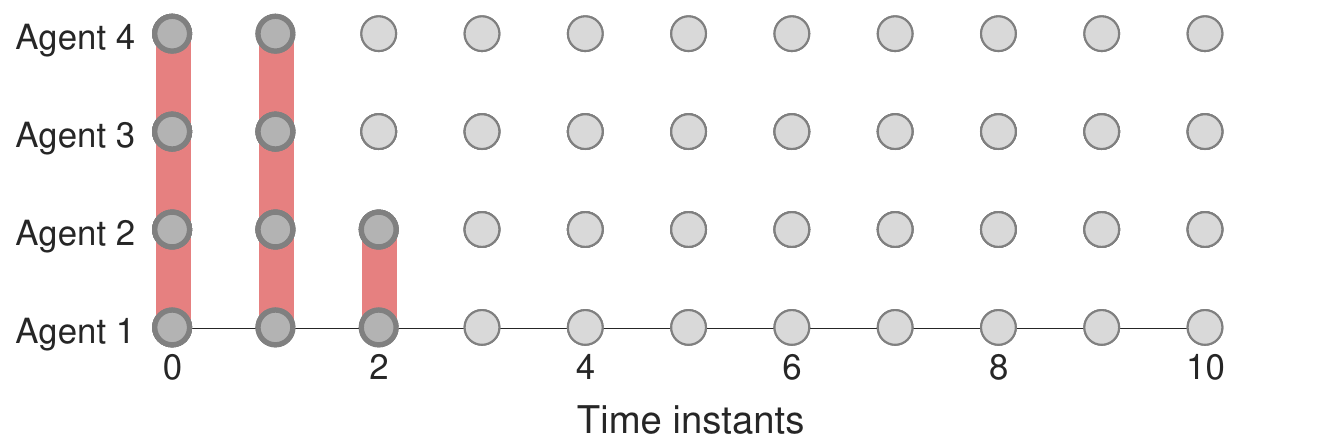}}
	\caption{Topologies imposed each time instant.}
	\label{top}
\end{figure}
Finally, to evaluate the controller's performance, we can compare the results obtained in terms of the cost function. In particular, to assess the control performance we consider $J_x=\sum_k (x_{\mathcal{N}}(k)^\mathrm{T} Q_x  x_{\mathcal{N}}(k) + u_{\mathcal{N}}(k)^\mathrm{T} R  u_{\mathcal{N}}(k))$, and, likewise, to evaluate the observer, we use $J_e=\sum_k e_{\mathcal{N}}(k)^\mathrm{T} Q_e e_{\mathcal{N}}(k)$. The value of $J_x$ for the coalitional controller is $1.48 \cdot 10^4$, while the centralized structure leads to $1.46 \cdot 10^4$ , and the decentralized ends at $2.27 \cdot 10^4$. For the observer, the coalitional controller reaches a cost $J_e$ equal to 200.6, which matches approximately the centralized case, however, for the decentralized observer it goes up to 211.63.

\section{Conclusions}
\label{sec:conclusion}
In this paper, a networked system of agents, which can be interconnected by both state and input is considered. For this class of systems a coalitional control and observation scheme is presented, in which the communication topologies are changed on-line. Switches between topologies are designed such that they contribute to minimizing a cost function which considers: control performance, control effort, state estimation performance, and communication costs. The designed control and observation scheme only requires the agents to communicate their own measurements and inputs within their coalitions. Inter-coalitions communication is only required for deciding the switchings of topologies.

The proposed control and observation scheme relies on a combination of ellipsoidal bounds on the state, input and state observation error, which are defined by solving linear matrix inequalities (LMIs). With these ellipsoidal bounds the control and observation scheme can be proven to stabilize the overall networked system. A simulation of a vehicle platoon control problem is shown to illustrate the proposed method.

Future research will extend these results to even more realistic scenarios, such as the case of Collaborative Adaptive Cruise Control (CACC). Furthermore, after the introduction of a state observer in this paper, introducing fault observers in the coalitional framework would be a very interesting step.

\bibliography{references}
\end{document}